\documentclass[11pt]{article}
\usepackage{natbib}
\usepackage{amsmath,amsthm,amsopn,amstext,amscd,amsfonts,amssymb,color,enumerate}

\def\Cov{\mathop{\rm Cov}\nolimits}
\def\Var{\mathop{\rm Var}\nolimits}

\def\V{\mathop{\rm \widehat{\textbf{V}}}\nolimits}
\def \E {\mathop{\rm E}\nolimits}
\def \P {\mathop{\rm P}\nolimits}
\def \build#1#2#3{\mathrel{\mathop{#1}\limits^{#2}_{#3}}}

\newtheorem{thm}{Theorem}[section]

\newtheorem{lem}{Lemma}[section]
\theoremstyle{definition}

\newtheorem{rem}{Remark}[section]

\newcommand{\ams}[2]
           {\begin{center}
            \begin{minipage}{5.25in}
            \small
            \noindent \textbf{Mathematics Subject Classification: }{\uppercase{#1}}
            \end{minipage}
            \end{center}
            \par\normalsize
           }

\newcommand{\keywords}[1]
           {\begin{center}
            \begin{minipage}{5.25in}
            \small
            \noindent \textbf{Key Words:}~{\textrm{#1}}
            \end{minipage}
            \end{center}
            \normalsize
           }

\title{\vspace{-2.5cm}\textbf{\Large Multivariate stratified sampling by
stochastic multiobjective optimisation}}
\author{
  \begin{normalsize}
  \begin{tabular}{c}
    \textbf{ Jos\'e A. D\'\i az-Garc\'\i a} \\
    Department of Statistics and Computation \\
    Universidad Aut\'onoma Agraria Antonio Narro \\
    25350 Buenavista, Saltillo, Coahuila, M\'EXICO. \\
    jadiaz@uaaan.mx \\
    and\\
    \textbf{Rogelio Ramos-Quiroga}\\
    Centro de Investigaci\'on en Matem\'aticas\\
    Department of Probability and Statistics\\
    Callej\'on de Jalisco s/n.\\
    36240 Guanajuato, M\'exico\\
    rramosq@cimat.mx\\
  \end{tabular}
  \end{normalsize}
}
\date{}

\begin{document}
  \maketitle
\begin{abstract}
\noindent This work considers the allocation problem for multivariate stratified random
sampling as a problem of integer non-linear stochastic multiobjective mathematical programming.
With this goal in mind the asymptotic distribution of the vector of sample variances is
studied. Two alternative approaches are suggested for solving the allocation problem for
multivariate stratified random sampling. An example is presented by applying the different
proposed techniques.
\end{abstract}


\ams{62D05, 90C15, 90C29, 90C10}

\section{Introduction}\label{sec1}

The theory of probabilistic sampling is one of the topics of statistical theory that is most
commonly used in many fields of scientific investigation. An effective survey technique in a
population can be seen as an appropriate extraction of useful data which provides meaningful
knowledge of the important aspects of the population. Stratified random sampling is one of the
classical methods of the theory of probabilistic sampling for obtaining such information. The
computation of the stratum sample size in stratified random sampling can be computed by diverse
approaches, but optimum allocation, according to some criteria, has been found to be a useful
approach. Optimum allocation has been stated as a non-linear mathematical programming problem
in which the objective function is the variance subject to a cost restriction, or vice versa.
Typically, this problem has been solved by using the Cauchy-Schwarz \citep{s54} inequality,
cited in \cite{coc77} or Lagrange's multiplier method, see \cite{sssa84}.

Classical survey theory considers a single decision variable or characteristic; for example, in
our case, univariate stratified random sampling studies one characteristic, the sample size and
its strata allocation, see \cite{coc77}, \cite{sssa84} and \cite{t97}. Moreover, in the context
of stratified random sampling, diverse multivariate approaches have been proposed whereby the
sample size and its allocation within strata take into account several characteristics, see
\cite{sssa84} and \cite{ad81}. A detailed study of this problem is given by \citet{dgu:08}.

In univariate and multivariate stratified random sampling, when the optimum allocation is
performed, and the cost function is the objective function, subject to certain variance
restrictions in the different characteristics, then the problem can be reduced to a classical
mathematical programming problem, and for this purpose there are two well-known approaches: see
\cite{coc77}, \cite{sssa84} for the univariate case and see \cite{ad81} for the multivariate
one, both from a deterministic point of view; and see \cite{dggt:07} and \cite{pre78}, for a
stochastic approach, respectively. Observe that, in the latter case (stochastic multivariate
case with the cost function as the objective function) the problem can be solved by using any
of the techniques presented in \cite{dggt:07}, among many others.

Alternatively, if the interest is to minimise  the variances subject to a cost function, or to
a given sample size, then several approaches can be adopted to solve this, for the univariate
case see \citet{coc77}, \citet{t97} and \cite{sssa84}; and see \cite{sssa84} and \citet{dgu:08}
for the multivariate case. \citet{dgu:08} show that all the previously published approaches in
this area are particular cases of the multiobjective mathematical programming technique, they
proposed a unified theory for solving the problem of optimum allocation in multivariate
stratified random sampling. Furthermore, \citet{dgu:08} propose the optimum allocation in
multivariate stratified random sampling as a nonlinear problem of integer matrix mathematical
programming constrained by a cost function or by a given sample size. Also, by defining a
particular vectorial function of the objective function of the matrix mathematical programming
problem, they show that the optimum allocation in multivariate stratified random sampling also
can be studied as a non-linear multiobjective integer mathematical programming problem.

A topic with very little work done is the problem of optimum allocation in multivariate
stratified random sampling when it is assumed that the variances in each stratum are unknown
(the most common case in practice), in which case it would be needed to estimate these
variances, see \citet{coc77}, \citet{t97} and \cite{sssa84}. Under this context, the set of
estimated variances for each stratum are random variables. This last situation has an important
consequence when the optimum allocation problem in stratified random sampling is stated as a
nonlinear mathematical programming. The univariate case for this setting was studied in detail
by \citet{dggt:07}. The present work owes some ideas to \citet{mel86} who studies the
asymptotic normality of the optimal solution in multivariate stratified random sampling by
taking into account the randomness of the set of estimated variances for each stratum.

The theory of \emph{Stochastic Mathematical Programming} or \emph{Stochastic Optimisation} is a
well established field of study that deals with the following general problem:
\begin{equation}\label{eq0}
  \begin{array}{c}
  \build{\min}{}{\mathbf{x}}
  \left(
  \begin{array}{c}
    h_{1}(\mathbf{x},\boldsymbol{\xi}) \\
    h_{2}(\mathbf{x},\boldsymbol{\xi}) \\
    \vdots \\
    h_{r}(\mathbf{x},\boldsymbol{\xi})
  \end{array}
  \right )\\
  \mbox{subject to}\\
  g_{j}(\mathbf{x},\boldsymbol{\xi})\geq 0, \ j=1,2\dots,s,\\
  \end{array}
\end{equation}
where $\mathbf{x}$ is $k$-dimensional and $\boldsymbol{\xi}$ is $m$-dimensional. If
$\mathbf{x}$ or $\boldsymbol{\xi}$ are random, then (\ref{eq0}) defines a multiobjective
stochastic mathematical programming problem, see \citet[p. 234]{p:95} and \citet{up:01}. As
shall be seen in the following sections, the optimum allocation problem in stratified random
sampling can be proposed as a stochastic multiobjective mathematical programming problem.

This paper formulates the optimum allocation in multivariate stratified random sampling as a
stochastic multiobjective integer mathematical programming problem constrained by a cost
function or by a given sample size. Section \ref{sec2} includes some notation and definitions on
multivariate stratified random sampling under two approaches. Section \ref{sec3} studies the
asymptotic normality of the vector of sample variances. Two alternative approaches for
solving the integer stochastic multiobjective mathematical programming problems are proposed in
Section \ref{sec4}. In Section \ref{sec5}, under the first approach, diverse particular
stochastic solutions are proposed. Similarly, several particular solutions are derived under
the second suggested approach, see Section \ref{sec6}. Section \ref{sec7} establish some
equivalent deterministic programs to the stochastic solution stated in Section \ref{sec6}.
Finally the techniques proposed are applied to an example of the literature, see Section
\ref{sec8}.

\section{Preliminary results on multivariate stratified random sampling}\label{sec2}

Assume a population of size $N$, divided into $H$ sub-populations (strata). We wish to find a
representative sample of size $n$ and an optimum allocation in the strata, meeting the following
requirements: i) to minimise the variance of the estimated mean subject to a budgetary
constraint; or ii) to minimise the cost subject to a constraint on the variances; this is the
classical problem in optimum allocation in univariate stratified sampling, see \cite{coc77},
\cite{sssa84} and \cite{t97}. However, if more than one characteristic (variable) is being considered
then the problem is known as optimum allocation in multivariate stratified sampling. For a
formal expression of the problem of optimum allocation in stratified sampling, consider the
following notation.

The subindex $h=1,2,\cdots,H$ denotes the stratum, $i=1,2,\cdots,N_{h} \mbox{ or } n_{h}$ the unit
within stratum $h$ and $j=1,2,\cdots,G$ denotes the characteristic (variable). Moreover:

\bigskip

\begin{footnotesize}
\begin{tabular}{ll}
    $N_{h}$ & Total number of units within stratum $h$\\
    $n_{h}$ &  Number of units from the sample in stratum $h$\\
    $y_{hi}^{j}$ &  Value obtained for the $i$-th unit in stratum $h$\\
    & of the $j$-th characteristic\\
\end{tabular}
\end{footnotesize}

\begin{footnotesize}
\begin{tabular}{ll}
    $\mathbf{n} = ({n}_{1},\cdots, {n}_{H})'$ & Vector of the number of units in the sample\\
    $\displaystyle{W_{h}} = \displaystyle{\frac{N_{h}}{N}}$ & Relative size of stratum  $h$\\[3ex]
    $\displaystyle{\overline{Y}_{h}^{j}} = \displaystyle{\frac{\displaystyle{
        \sum_{i=1}^{N_{h}}}y_{hi}^{j}}{N_{h}}}$ & Population mean in stratum $h$ of the $j$-th characteristic\\[3ex]
    $\overline{\mathbf{Y}}_{h}=  (\overline{Y}_{h}^{1}, \cdots,\overline{Y}_{h}^{G})'$
    & Population mean vector in stratum $h$ \\[1ex]
    $\displaystyle{\overline{y}_{h}^{j}}=\displaystyle{\frac{\displaystyle{
        \sum_{i=1}^{n_{h}}}y_{hi}^{j}}{n_{h}}}$ & Sample mean in stratum $h$ of the $j$-th characteristic\\[3ex]
    $\overline{\mathbf{y}}_{h}=  (\overline{y}_{h}^{1}, \cdots,\overline{y}_{h}^{G})'$
    & Sample mean vector in stratum $h$ \\[1ex]
    $\displaystyle{\overline{y}_{_{ST}}^{j}=  \sum_{h=1}^{H}W_{h}\overline{y}_{h}^{j}}$
    & \hspace{-.2cm}\begin{tabular}{l}
        Estimator of the population mean in multivariate \\
        stratified sampling for the $j$-th characteristic \\
      \end{tabular}\\
    $\overline{\mathbf{y}}_{_{ST}}=(\overline{y}_{_{ST}}^{1},\cdots,\overline{y}_{_{ST}}^{G})'$
    & \hspace{-.2cm}\begin{tabular}{l}
    Estimator of the population mean vector in\\
    multivariate stratified sampling\\
    \end{tabular}\\
    $S_{hj}^{2}$ & Population variance in stratum $h$\\
    & $S_{hj}^{2} = \displaystyle{\frac{\displaystyle{
        \sum_{i=1}^{N_{h}}}(y_{hi}^{j}-\overline{y}_{h}^{j})^2}{N_{h}}}$\\
\end{tabular}

\begin{tabular}{ll}
    $s_{hj}^{2}$ & Sample variance in stratum $h$\\
    & $s_{hj}^{2} = \displaystyle{\frac{\displaystyle{
        \sum_{i=1}^{n_{h}}}(y_{hi}^{j}-\overline{y}_{h}^{j})^2}{{n_{h}}-1}}$\\
    $\mathbb{S}_{h} = \left(S_{h1}^{2}, \cdots, S_{hG}^{2}\right)'$ & Vector of population variances in stratum
    $h$\\[2ex]
     $\mathfrak{s}_{h} = \left(s_{h1}^{2}, \cdots, s_{hG}^{2}\right)'$ & Vector of variance estimators in stratum
     $h$\\[2ex]
    $\widehat{\Var}(\overline{y}_{_{ST}}^{j})$ & Estimated variance of
    $\overline{y}_{_{ST}}^{j}$\\[2ex]
    & $ \widehat{\Var}(\overline{y}_{_{ST}}^{j})= \displaystyle{\sum_{h=1}^{H}\frac{{{W_{h}}^{2}}
        s_{hj}^{2}}{n_{h}} - \sum_{h=1}^{H} \frac{{W_{h}}s_{hj}^{2}}{N}}$ \\[2ex]
       $c_{h}$ & Cost per $G$-dimensional sampling unit in stratum $h$ and let\\
       & $\mathbf{c} = (c_{1}, \cdots, c_{G})'$.\\[2ex]
\end{tabular}
\end{footnotesize}

\noindent Define $\V_{u}(\overline{\mathbf{y}}_{_{ST}}) \in \Re^{G}$, as
$\V_{u}(\overline{\mathbf{y}}_{_{ST}}) = \left(\widehat{\Var}(\overline{y}_{_{ST}}^{1}), \cdots,
\widehat{\Var} (\overline{y}_{_{ST}}^{G})\right)'$, where if $\mathbf{a} \in \Re^{G}$,
$\mathbf{a}'$ denotes the transpose of $\mathbf{a}$. Finally, observe that
$\V_{u}(\overline{\mathbf{y}}_{_{ST}})$ is a function of $\mathbf{n}$, then it shall be written
either $\V_{u}(\overline{\mathbf{y}}_{_{ST}})$ or
$\V_{u}\left(\overline{\mathbf{y}}_{_{ST}}(\mathbf{n})\right)$. Similarly, it shall be written
$\widehat{\Var}(\overline{y}_{_{ST}}^{j})$ or
$\widehat{\Var}(\overline{y}_{_{ST}}^{j}(\mathbf{n}))$.

\section{Limiting distribution of the sample variance}\label{sec3}

In this section the asymptotic distribution of the estimator of the vector of
variances $\mathfrak{s}_{h}$ is studied. With this aim in mind, the multivariate version of
H\'ajek's theorem in the context of sampling theory is presented, see \citet{h:61}.
In what follows, from Lemma \ref{lemma1} through Theorem \ref{teo1}, asymptotic results are stated
for a single stratum. The notation $N_{\nu}$ and $n_{\nu}$ denote the size of a generic stratum
and the size of a simple random sample from that stratum.

\begin{lem}\label{lemma1}
Let $\boldsymbol{\varrho}_{\nu}$ be a $G \times 1$ random vector defined as
\begin{equation}\label{varrho}
    \boldsymbol{\varrho}_{\nu} =
   \left [
   \begin{array}{c}
     \varrho_{\nu}^{1} \\
     \vdots \\
     \varrho_{\nu}^{G}
   \end{array}
   \right] =
   \frac{1}{n_{\nu}-1}
   \left(
      \sum_{i = 1}^{n_{\nu}}\left(y_{\nu i}^{j}- \overline{Y}_{\nu}^{j}\right)^{2}
   \right)_{ \ j = 1, \dots, G}.
\end{equation}
Assume that for $\boldsymbol{\lambda} = (\lambda_{1}, \dots, \lambda_{G})'$, any vector of
constants,%
{\small
\begin{equation}\label{shc}
  \boldsymbol{\lambda}'\left(\mathbf{M}_{\nu}^{4} - \mathbb{S}_{\nu}
  \mathbb{S}'_{\nu}\right) \boldsymbol{\lambda} \geq \epsilon \build{\max}{}{1 \leq j \leq
  G} \left[\lambda_{j}^{2} \mathbf{e}_{G}^{j '}\left(\mathbf{M}_{\nu}^{4} - \mathbb{S}_{\nu}
  \mathbb{S}'_{\nu}\right) \mathbf{e}_{G}^{j}\right],
\end{equation}}
where $\mathbf{e}_{G}^{j } = (0, \dots, 0, 1, 0, \dots, 0)'$ is the $j$-th vector of the
canonical base of $\Re^{G}$, $\epsilon > 0$ and independent of $\nu > 1$ and
$$
  \mathbf{M}_{\nu}^{4} = \frac{1}{N_{\nu}}\left(\sum_{i = 1}^{N_{\nu}}
  \left(y_{\nu i}^{k} - \overline{Y}_{\nu}^{k}\right)^{2}\left(y_{\nu i}^{l} - \overline{Y}_{\nu}^{l}\right)^{2}
  \right)_{k,l = 1, \dots, G}.
$$
Suppose that $n_{\nu}\rightarrow \infty$, $N_{\nu} - n_{\nu}
\rightarrow \infty$, $N_{\nu}\rightarrow \infty$, and that, for all $j = 1,\dots,G$,%
\begin{equation}\label{hcas}
  \left[\build{\lim}{}{\nu \rightarrow \infty}\left(\frac{n_{\nu}}{N_{\nu}}\right) = 0\right] \Rightarrow
  \build{\lim}{}{\nu \rightarrow \infty} \frac{\build{\max}{}{1 \leq i_{1} < \cdots < i_{n_{\nu}}\leq N_{\nu}}
  \displaystyle\sum_{\beta = 1}^{n_{\nu}}\left[\left(y_{\nu i_{\beta}}^{j} - \overline{Y}_{\nu}^{j}\right)^{2} - S_{\nu j}^{2}\right]^{2}}
  {N_{\nu}\left[m_{\nu j}^{4}-\left(S_{\nu j}^{2}\right)^{2}\right]} = 0,
\end{equation}
where
$$
  m_{\nu j}^{4} =\frac{1}{N_{\nu}}\sum_{i = 1}^{N_{\nu}}\left(y_{\nu i}^{j} -
  \overline{y}_{\nu}^{j}\right)^{4}.
$$
Then, $\boldsymbol{\varrho}_{\nu}$ is asymptotically normally distributed as
$$
  \boldsymbol{\varrho}_{\nu} \build{\rightarrow}{d}{} \mathcal{N}_{G}(\E(\boldsymbol{\varrho}_{\nu}),
  \Cov(\boldsymbol{\varrho}_{\nu})),
$$
with
\begin{equation}\label{mXi}
    \E(\boldsymbol{\varrho}_{\nu}) = \frac{n_{\nu}}{n_{\nu}-1}\mathbb{S}_{\nu},
\end{equation}
and
\begin{equation}\label{cmXi}
    \Cov(\boldsymbol{\varrho}_{\nu}) = \frac{n_{\nu}}{(n_{\nu} - 1)^{2}}\left(\mathbf{M}_{\nu}^{4}
     - \mathbb{S}_{\nu}\mathbb{S}'_{\nu}\right).
\end{equation}
$n_{\nu}$ is the sample size for a simple random sample from the $\nu$-th population of size
$N_{\nu}$.
\end{lem}

\begin{rem}\label{remark1}
Let $\boldsymbol{\varrho}_{\nu}$ as in (\ref{varrho}), taking $m = G$ and $a_{\nu i} =
\left(y_{\nu i}^{j}- \overline{Y}_{\nu}^{j}\right)^{2}$ in \citet{h:61}, it is obtained that:
\begin{description}
  \item[i)] $ \boldsymbol{\varrho}_{\nu}$ can be expressed as
  $$
    \boldsymbol{\varrho}_{\nu} = \left(\sum_{i =1}^{N_{\nu}} b_{\nu i} a_{\nu R_{\nu i}}\right)_{ \ j = 1, \dots, G}.
  $$
  with $b's$ fixed, furthermore $b_{\nu 1} = \cdots = b_{\nu n_{\nu}} = 1/(n_{\nu}-1)$, $b_{\nu n_{\nu}+1} = \cdots = b_{\nu N_{\nu}} =
  0$. Then
  $$
    \build{\lim}{}{\nu \rightarrow \infty} \frac{\build{\max}{}{1 \leq j \leq N_{\nu}} \left(b_{\nu j} -
    \overline{b}_{\nu}\right)^{2}}{\displaystyle\sum_{i=1}^{N_{\nu}}\left(b_{\nu j} -
    \overline{b}_{\nu}\right)^{2}} = 0, \quad \mbox{ where } \quad \overline{b}_{\nu} = \frac{1}{N_{\nu}}\sum_{i = 1}^{N_{\nu}}b_{\nu i}
  $$
  holds if $n_{\nu}\rightarrow \infty$, $N_{\nu} - n_{\nu} \rightarrow \infty$.

  \item[ii)] $\overline{a}_{\nu}$ is
  \begin{eqnarray*}
    \overline{a}_{\nu} &=& \frac{1}{N_{\nu}} \sum_{i = 1}^{N_{\nu}}a_{\nu i}\\
     &=& \frac{1}{N_{\nu}} \sum_{i = 1}^{N_{\nu}} \left(y_{\nu i}^{j}- \overline{Y}_{\nu}^{j}\right)^{2}\\
     &=& S_{\nu j}^{2}
  \end{eqnarray*}
  \item[iii)] From (7.2) in \citet{h:61}
  \begin{equation}\label{ohc}
    \sum_{i = 1}^{N_{\nu}}\left[\sum_{j = 1}^{k}\lambda_{j} (a_{\nu i}^{j} -
    a_{\nu}^{j})\right]^{2} \geq \epsilon \build{\max}{}{1 \leq j \leq k}
    \left[\lambda_{j}^{2}\sum_{i = 1}^{N_{\nu}} (a_{\nu i}^{j} -
    a_{\nu}^{j})^{2}\right].
  \end{equation}
  In the context of sampling theory, the right side in (\ref{ohc}) can be written as
  \begin{eqnarray*}
    \sum_{i = 1}^{N_{\nu}}\left[\sum_{j = 1}^{k}\lambda_{j} (a_{\nu i}^{j} - a_{\nu}^{j})\right]^{2}
     &=& \sum_{i = 1}^{N_{\nu}}\left\{\sum_{j = 1}^{k}\lambda_{j} \left[\left(y_{\nu i}^{j}-
     \overline{Y}_{\nu}^{j}\right)^{2} - S_{\nu j}^{2}\right]\right\}^{2} \\
  \end{eqnarray*}

  \vspace{-1.5cm}

  \begin{eqnarray}
     \quad &=& \sum_{i = 1}^{N_{\nu}}\left\{\boldsymbol{\lambda}'
     \left[
     \begin{array}{c}
       \left(y_{\nu i}^{1}- \overline{Y}_{\nu}^{1}\right)^{2} - S_{\nu 1}^{2} \\
       \vdots \\
       \left(y_{\nu i}^{G}- \overline{Y}_{\nu}^{G}\right)^{2} - S_{\nu G}^{2}
     \end{array}
     \right]\right\}^{2} \nonumber\\
     &=& \sum_{i = 1}^{N_{\nu}}\boldsymbol{\lambda}'
     \left[
     \begin{array}{c}
       \left(y_{\nu i}^{1}- \overline{Y}_{\nu}^{1}\right)^{2} - S_{\nu 1}^{2} \\
       \vdots \\
       \left(y_{\nu i}^{G}- \overline{Y}_{\nu}^{G}\right)^{2} - S_{\nu G}^{2}
     \end{array}
     \right]\nonumber\\
     & & \qquad\quad \left[\left(\left(y_{\nu i}^{1}- \overline{Y}_{\nu}^{1}\right)^{2} - S_{\nu 1}^{2}\right),\dots,
     \left(\left(y_{\nu i}^{G}- \overline{Y}_{\nu}^{G}\right)^{2} - S_{\nu G}^{2}\right)\right]
     \boldsymbol{\lambda}\nonumber\\ \label{lshc}
     &=& N_{\nu}\boldsymbol{\lambda}'\left[m_{\nu_{k l}}^{4} - S_{\nu_{k}}^{2}
     S_{\nu_{l}}^{2}\right]_{ \ k,l = 1, \dots, G} \boldsymbol{\lambda}\nonumber\\
     &=& N_{\nu}\boldsymbol{\lambda}'\left(\mathbf{M}_{\nu}^{4} - \mathbb{S}_{\nu}
     \mathbb{S}'_{\nu}\right) \boldsymbol{\lambda},
  \end{eqnarray}
  where $\mathbf{M}_{\nu}^{4}$ is%
  {\small
  \begin{equation}\label{m4}
     = \frac{1}{N_{\nu}}\left[\sum_{i = 1}^{N_{\nu}}
    \left(y_{\nu i}^{k}- \overline{Y}_{\nu}^{k}\right)^{2}\left(y_{\nu i}^{l}- \overline{Y}_{\nu}^{l}\right)^{2}
    \right]_{ \ k,l = 1, \dots, G},
  \end{equation}}
  Similarly the right side of (\ref{ohc}) is
  \begin{eqnarray*}
  \hspace{-1cm}
    \lambda_{j}^{2}\sum_{i = 1}^{N_{\nu}} (a_{\nu i}^{j} - a_{\nu}^{j})^{2}
     &=& \sum_{i = 1}^{N_{\nu}}\left\{\boldsymbol{\lambda}' \mathbf{e}_{G}^{j}\mathbf{e}_{G}^{j'}
     \left[\left(y_{\nu i}^{j}- \overline{Y}_{\nu}^{j}\right)^{2} - S_{\nu j}^{2}\right]\right\}^{2} \\
     &=& \lambda_{j}^{2}\sum_{i = 1}^{N_{\nu}}\left\{\mathbf{e}_{G}^{j'}
     \left[\left(y_{\nu i}^{j}- \overline{Y}_{\nu}^{j}\right)^{2} - S_{\nu j}^{2}\right]\right\}^{2}.
  \end{eqnarray*}
  Then, proceeding as in 3.,
  \begin{equation}\label{rshc}
    \lambda_{j}^{2}\sum_{i = 1}^{N_{\nu}} (a_{\nu i}^{j} - a_{\nu}^{j})^{2} =
    N_{\nu} \lambda_{j}^{2} \mathbf{e}_{k}^{j '}\left(\mathbf{M}_{\nu}^{4} -
    \mathbb{S}_{\nu} \mathbb{S}'_{\nu}\right) \mathbf{e}_{k}^{j}.
  \end{equation}
  Therefore, from (\ref{lshc}) and (\ref{rshc}), (\ref{shc}) is established.
  \item[iv)] The expression for (\ref{hcas}) is found analogous to the procedure described in item 3.
  \item[v)] Finally,
  \begin{eqnarray*}
    \E(\boldsymbol{\varrho}) &=& \frac{1}{n_{\nu}-1}\E\left(\sum_{i = 1}^{n_{\nu}}\left(y_{\nu i}^{j}
      - \overline{Y}_{\nu}^{j}\right)^{2}\right)_{ \ j = 1, \dots, G} \\
     &=& \frac{1}{n_{\nu}-1}\left( \sum_{i = 1}^{n_{\nu}}\E\left(y_{\nu i}^{j}
      - \overline{Y}_{\nu}^{j}\right)^{2}\right)_{ \ j = 1, \dots, G}  \\
     &=& \frac{1}{n_{\nu}-1}\left( \sum_{i = 1}^{n_{\nu}}S_{\nu j}^{2}\right)_{ \ j = 1, \dots, G}  \\
     &=& \frac{n_{\nu}}{n_{\nu}-1} \mathbb{S}_{\nu}\\
  \end{eqnarray*}
  Similarly, by independence
  $$\hspace{-2cm}
    \Cov(\boldsymbol{\varrho}) = \frac{1}{(n_{\nu}-1)^{2}} \Cov\left(\sum_{i = 1}^{n_{\nu}}\left(y_{\nu i}^{j}
      - \overline{Y}_{\nu}^{j}\right)^{2}\right)_{ \ j = 1, \dots, G}
  $$

  \vspace{-.75cm}

  {\small
  \begin{eqnarray*}
     \quad &=& \frac{1}{(n_{\nu}-1)^{2}} \left\{\E\left(\sum_{i = 1}^{n_{\nu}}\left(y_{\nu i}^{j}
      - \overline{Y}_{\nu}^{j}\right)^{2}\right)_{ \ j = 1, \dots, G} \left(\sum_{i = 1}^{n_{\nu}}\left(y_{\nu i}^{j}
      - \overline{Y}_{\nu}^{j}\right)^{2}\right)'_{ \ j = 1, \dots, G}\right.\\
      & & \left . \ - \E\left(\sum_{i = 1}^{n_{\nu}}\left(y_{\nu i}^{j}
      - \overline{Y}_{\nu}^{j}\right)^{2}\right)_{ \ j = 1, \dots, G}\E\left(\sum_{i = 1}^{n_{\nu}}\left(y_{\nu i}^{j}
      - \overline{Y}_{\nu}^{j}\right)^{2}\right)'_{ \ j = 1, \dots, G}\right\}\\
     &=& \frac{n_{\nu}}{(n_{\nu}-1)^{2}}  \left(\mathbf{M}_{\nu}^{4} - \mathbb{S}_{\nu}\mathbb{S}'_{\nu}\right),
  \end{eqnarray*}}
  where $\mathbf{M}_{\nu}^{4}$ is defined in (\ref{m4}).
\end{description}
\end{rem}

\begin{thm}\label{teo1}
Under assumptions in Lemma \ref{lemma1}, the sequence of vector of sample variances
$\mathfrak{s}_{\nu}$ are such that $\mathfrak{s}_{\nu}$ has an asymptotical normal with
asymptotic mean and covariance matrix given by (\ref{mXi}) and (\ref{cmXi}), respectively.
\end{thm}
\begin{proof}
This follows immediately from Lemma \ref{lemma1}, only observe that
\begin{eqnarray*}
  \mathfrak{s}_{\nu} &=&  \frac{1}{n_{\nu}-1}\left(\sum_{i = 1}^{n_{\nu}}\left(y_{\nu i}^{j}-
   \overline{y}_{\nu}^{j}\right)^{2}\right)_{j =1,\dots,G} \\
   &=& \boldsymbol{\varrho} - \frac{n_{\nu}}{n_{\nu}-1} \left((\overline{y}_{\nu}^{j}-
   \overline{Y}_{\nu}^{j})^{2}\right)_{j =1,\dots,G}
\end{eqnarray*}
where
$$
  \frac{n_{\nu}}{n_{\nu}-1} \rightarrow 1 \quad \mbox{ and } \quad (\overline{y}_{\nu}^{j}-
   \overline{Y}_{\nu}^{j})^{2} \rightarrow 0 \quad \mbox{in
   probability}, j =1,\dots,G. \qquad\mbox{\qed}
$$
\end{proof}

As a direct consequence of Theorem \ref{teo1} it is obtained:

\begin{thm}
Let $\V_{u}(\overline{\mathbf{y}}_{_{ST}})$ be the estimator of the vector of variances of
$\overline{\mathbf{y}}_{ST}$, then
$$
  \V_{u}(\overline{\mathbf{y}}_{_{ST}}) = \sum_{h=1}^{H}\left(\frac{{{W_{h}}^{2}}}{n_{h}} -
  \frac{{W_{h}}}{N} \right)\mathfrak{s}_{h}
$$
is asymptotically normal distributed; moreover
$$
  \V_{u}(\overline{\mathbf{y}}_{_{ST}}) \build{\rightarrow}{d}{} \mathcal{N}_{k}
  \left(\E\left(\V_{u}(\overline{\mathbf{y}}_{_{ST}})\right),
  \Cov\left(\V_{u}(\overline{\mathbf{y}}_{_{ST}})\right)\right),
$$
where
\begin{equation}\label{ecyst}
    \E\left(\V_{u}(\overline{\mathbf{y}}_{_{ST}})\right) =  \sum_{h=1}^{H}\left(
     \frac{{{W_{h}}^{2}}}{n_{h}} - \frac{{W_{h}}}{N} \right) \frac{n_{h}}{n_{h}-1}\mathbb{S}_{h},
\end{equation}
\begin{equation}\label{ccyst}
    \Cov\left(\V_{u}(\overline{\mathbf{y}}_{_{ST}})\right) =\sum_{h=1}^{H}\left(
    \frac{{{W_{h}}^{2}}}{n_{h}} - \frac{{W_{h}}}{N} \right) \frac{n_{h}}{(n_{h}-1)^{2}}
    \left(\mathbf{M}_{h}^{4} - \mathbb{S}_{h}\mathbb{S}'_{h}\right),
\end{equation}
and
$$
  \mathbf{M}_{h}^{4} = \frac{1}{N_{h}}\left[\sum_{i = 1}^{N_{h}}
    \left(y_{h i}^{k}- \overline{Y}_{h}^{k}\right)^{2}\left(y_{h i}^{l}- \overline{Y}_{h}^{l}\right)^{2}
    \right]_{ \ k,l = 1, \dots, G},
$$
$n_{\nu}$ is the sample size for a simple random sample from the $\nu$-th population of size
$N_{\nu}$.
\end{thm}

Observe that the asymptotic means and covariance matrices of the asymptotic normal
distributions of $ \mathfrak{s}_{h}$ and $\V_{u}(\overline{\mathbf{y}}_{_{ST}})$ are in terms
of the population parameters $\overline{\mathbf{Y}}_{h}$, $\mathbb{S}_{h}$ and
$\mathbf{M}_{h}^{4}$; then, from \citet{r:73}, approximations of the asymptotic distribution
are obtained making the following substitutions
\begin{equation}\label{sus}
    \mathbb{S}_{h} \rightarrow \mathfrak{s}_{h} , \quad \mbox{ and } \quad \mathbf{M}_{h}^{4} \rightarrow \mathbf{m}_{h}^{4}
\end{equation}
where
$$
  \mathbf{m}_{h}^{4} = (\widehat{m}_{h_{kl}}^{4})_{k,l = 1, \dots, G} = \frac{1}{n_{h}}\left(\sum_{i = 1}^{n_{h}}
  \left(y_{h i}^{k} - \overline{y}_{h}^{k}\right)^{2}\left(y_{h i}^{l} - \overline{y}_{h}^{l}\right)^{2}
  \right)_{k,l = 1, \dots, G}.
$$

\section{Optimum allocation in multivariate stratified random sampling via multiobjective
stochastic mathematical programming}\label{sec4}

When the variances are the objective functions, subject to certain cost function,  the optimum
allocation in multivariate stratified random sampling can be expressed as the following
deterministic multiobjective mathematical programming
\begin{equation}\label{equiv2}
  \begin{array}{c}
    \build{\min}{}{\mathbf{n}}\V_{u}(\overline{\mathbf{y}}_{_{ST}})= \
    \build{\min}{}{\mathbf{n}}
    \left(%
    \begin{array}{c}
      \widehat{\Var}(\overline{y}_{_{ST}}^{1}) \\
      \vdots\\
      \widehat{\Var}(\overline{y}_{_{ST}}^{G}) \\
    \end{array}%
    \right)\\
    \mbox{subject to}\\
    \mathbf{c}'\mathbf{n} + c_{0} = C  \\
    2\leq n_{h}\leq N_{h}, \ \ h=1,2,\dots, H\\
    n_{h}\in \mathbb{N},
  \end{array}
  \end{equation}
which has been studied in detail by \citet{dgu:08}.

Observing that $\V_{u}(\overline{\mathbf{y}}_{_{ST}})$ is a random variable then, as in (\ref{eq0}),
 the optimum allocation via stochastic mathematical programming can be stated as the
following stochastic multiobjective mathematical programming problem
\begin{equation}\label{smmp}
  \begin{array}{c}
    \build{\min}{}{\mathbf{n}}\V_{u}(\overline{\mathbf{y}}_{_{ST}})= \
    \build{\min}{}{\mathbf{n}}
    \left(%
    \begin{array}{c}
      \widehat{\Var}(\overline{y}_{_{ST}}^{1}) \\
      \vdots\\
      \widehat{\Var}(\overline{y}_{_{ST}}^{G}) \\
    \end{array}%
    \right)\\
    \mbox{subject to}\\
    \mathbf{c}'\mathbf{n} + c_{0} = C  \\
    2\leq n_{h}\leq N_{h}, \ \ h=1,2,\dots, H\\
    \V_{u}(\overline{\mathbf{y}}_{_{ST}}) \build{\rightarrow}{d}{} \mathcal{N}_{G}
    \left(\E\left(\V_{u}(\overline{\mathbf{y}}_{_{ST}})\right),
    \Cov\left(\V_{u}(\overline{\mathbf{y}}_{_{ST}})\right)\right)\\
    n_{h}\in \mathbb{N},
  \end{array}
\end{equation}
where $\E\left(\V_{u}(\overline{\mathbf{y}}_{_{ST}})\right)$ and
$\Cov\left(\V_{u}(\overline{\mathbf{y}}_{_{ST}})\right)$ are given by (\ref{ecyst}) and
(\ref{ccyst}) respectively.

In the statistical context, two approaches can be proposed to solve problems of stochastic
multiobjective mathematical programming of the type (\ref{smmp}):

\begin{enumerate}
  \item First, proposing a solution to the multiobjective mathematical programming problem,
  by converting it to an uniobjective mathematical programming problem. Then this uniobjective
  mathematical programming problem is solved as a stochastic programming problem  by applying
  any of the techniques in the area. This approach has been applied in other areas of the Statistics,
  see \citet{kc:87} among others.

  \item Given that methodologies exist, problems of the type (\ref{smmp}) can be solved by applying directly the
  techniques of stochastic multiobjective mathematical programming, see \citet{sm:84} among
  others.
\end{enumerate}

\section{First approach}\label{sec5}

For this approach, the literature of multiobjective mathematical programming
offers plenty of tools, see \citet{dgu:08}, \citet{rio89}, \citet{m99} and \citet{s86}.
Here, the techniques of value function and goal
programming are considered.  Under the value
function approach, (\ref{smmp}) is formulated as
\begin{equation}\label{fv}
  \begin{array}{c}
    \build{\min}{}{\mathbf{n}}\phi\left(\V_{u}(\overline{\mathbf{y}}_{_{ST}})\right),\\
    \mbox{subject to}\\
    \mathbf{c}'\mathbf{n} + c_{0}=C \\
    2\leq n_{h}\leq N_{h}, \ \ h=1,2,\dots, H\\
    \V_{u}(\overline{\mathbf{y}}_{_{ST}}) \build{\rightarrow}{d}{} \mathcal{N}_{G}
    \left(\E\left(\V_{u}(\overline{\mathbf{y}}_{_{ST}})\right),
    \Cov\left(\V_{u}(\overline{\mathbf{y}}_{_{ST}})\right)\right)\\
    n_{h}\in \mathbb{N},
  \end{array}
\end{equation}
where $\mathbf{\phi}(\cdot)$ is a value function\footnote{A value function is a function $\phi:
\Re^{H} \rightarrow \Re$ such that $\min \V_{u}(\overline{\mathbf{y}}_{_{ST}}(\mathbf{n}^{*}))
< \min \V_{u}(\overline{\mathbf{y}}_{_{ST}}(\mathbf{n}_{1})) \Leftrightarrow
\phi(\V_{u}(\overline{\mathbf{y}}_{_{ST}}(\mathbf{n}^{*}))) <
\phi(\V_{u}(\overline{\mathbf{y}}_{_{ST}}(\mathbf{n}_{1}))), \quad \mathbf{n}*\neq
\mathbf{n}_{1}$.} that summarises the importance of each of the variances of the $G$
characteristics.

Similarly, in terms of goal programming (\ref{smmp}) is stated as
\begin{equation}\label{fv1}
  \begin{array}{c}
    \build{\min}{}{\mathbf{n}}F(\mathbf{n}, \mathbb{S}) = \; \build{\min}{}{\mathbf{x}}\displaystyle\sum_{j=1}^{p}w_{j}(d_{j}^{+} + d_{j}^{-}) \\
    \mbox{subject to} \\
    \widehat{\Var}(\overline{y}_{_{ST}}^{j}(\mathbf{n})) + d_{j}^{+} - d_{j}^{-} = t_{j}, \quad j = 1,\dots,G, \\
    \mathbf{c}'\mathbf{n}+ c_{0}=C \\
    2\leq n_{h}\leq N_{h}, \ \ h=1,2,\dots, H\\
    \V_{u}(\overline{\mathbf{y}}_{_{ST}}) \build{\rightarrow}{d}{} \mathcal{N}_{G}
  \left(\E\left(\V_{u}(\overline{\mathbf{y}}_{_{ST}})\right),
  \Cov\left(\V_{u}(\overline{\mathbf{y}}_{_{ST}})\right)\right)\\
    n_{h}\in \mathbb{N},
  \end{array}
\end{equation}
with
$$
  d_{j}^{+} = \frac{1}{2} \left( \left|\widehat{\Var}(\overline{y}_{_{ST}}^{j}(\mathbf{n})) - t_{j} \right| +
  \left(\widehat{\Var}(\overline{y}_{_{ST}}^{j}(\mathbf{n})) -   t_{j}\right) \right),
$$
$$
  d_{j}^{-} = \frac{1}{2} \left( \left|\widehat{\Var}(\overline{y}_{_{ST}}^{j}(\mathbf{n})) -
  t_{j} \right| - \left(\widehat{\Var}(\overline{y}_{_{ST}}^{j}(\mathbf{n})) -  t_{j}\right)\right),
$$
where the $t_{j}$'s are the target values for the objectives
$\widehat{\Var}(\overline{y}_{_{ST}}^{j})$, $j=1,\dots, G$.

Note that, now, (\ref{fv}) and (\ref{fv1}) are stochastic uniobjective mathematical programming problems,
then, any technique from the area of stochastic programming could be applied. For
example:

\subsection{Modified expected value solution, $E$-model and $V$-model}

Point $\mathbf{n} \in \mathbb{N}^{H}$ is the expected modified value solution to (\ref{fv}) if
it is an efficient solution in the \textbf{Pareto}\footnote{For matrix mathematical programming
problems, there rarely exists a point $\mathbf{f}^{*}(\mathbf{n})$ which can be considered as
the minimum. Alternatively, it said that $\mathbf{f}^{*}(\mathbf{n})$ is a \textit{Pareto
point} of $\mathbf{f}(\mathbf{n}) = (f_{1}(\mathbf{n}), \dots, f_{G}(\mathbf{n}))'$, if there
is no other point $\mathbf{f}^{1}(\mathbf{n^{*}})$ such that $\mathbf{f}^{1}(\mathbf{n}) \leq
\mathbf{f}^{*}(\mathbf{n})$, i.e. for all $j$, $f_{j}(\mathbf{n^1}) \leq f_{j}(\mathbf{n^*})$
and $\mathbf{f}^{1}(\mathbf{n}) \neq \mathbf{f}^{*}(\mathbf{n})$.} sense to the following
deterministic uniobjetive mathematical programming problem
\begin{equation}\label{emsma}
    \begin{array}{c}
    \build{\min}{}{\mathbf{n}} k_{1}\E\left(\phi\left(\V_{u}(\overline{\mathbf{y}}_{_{ST}})\right)\right)
            + k_{2}\sqrt{\Var\left(\phi\left(\V_{u}(\overline{\mathbf{y}}_{_{ST}})\right)\right)}\\
    \mbox{subject to}\\
    \mathbf{c}'\mathbf{n}+ c_{0}=C \\
    2\leq n_{h}\leq N_{h}, \ \ h=1,2,\dots, H\\
    n_{h}\in \mathbb{N},
  \end{array}
\end{equation}
and for (\ref{fv1})
\begin{equation}\label{emsma1}
    \begin{array}{c}
    \build{\min}{}{\mathbf{n}} k_{1}\E\left(F(\mathbf{n}, \mathbb{S})\right)
            + k_{2}\sqrt{\Var\left(F(\mathbf{n}, \mathbb{S})\right)}\\
    \mbox{subject to}\\
    \widehat{\Var}(\overline{y}_{_{ST}}^{j}(\mathbf{n})) + d_{j}^{+} - d_{j}^{-} = t_{j}, \quad j = 1,\dots,G, \\
    \mathbf{c}'\mathbf{n}+ c_{0}=C \\
    2\leq n_{h}\leq N_{h}, \ \ h=1,2,\dots, H\\
    n_{h}\in \mathbb{N},
  \end{array}
\end{equation}

Here $k_{1}$ and $k_{2}$ are non negative constants, and their values show the relative
importance of the expectation and the variance of
$\phi\left(\V_{u}(\overline{\mathbf{y}}_{_{ST}})\right)$ and $F(\mathbf{n}, \mathbb{S})$. Some
authors suggest that $k_{1} + k_{2} =1$, see \citet[p. 599]{rao78}. Observe that if we take
$k_{1} = 1$ and $k_{2} = 0$ in (\ref{emsma}), the resulting method is known as the E-model.
Alternatively, if $k_{1} = 0$ and $k_{2} = 1$, the method is termed the V-model, see
\citet{chc:63}, \citet{p:95} and \citet{up:01}.

\subsection{Minimum risk solution of aspiration level $\tau$, $\P$-model}

Point $\mathbf{n} \in \mathbb{N}^{H}$ is a minimum risk solution at the aspiration level $\tau$
to the problem (\ref{fv}) (also termed P-model by \citet{chc:63}) if it is an efficient
solution in the Pareto sense of the uniobjetive stochastic mathematical programming problem
\begin{equation}\label{mrsma}
    \begin{array}{c}
    \build{\min}{}{\mathbf{n}} \P\left(\phi\left(\V_{u}(\overline{\mathbf{y}}_{_{ST}})\right) \leq \tau\right)\\
    \mbox{subject to}\\
    \mathbf{c}'\mathbf{n}+ c_{0}=C \\
    2\leq n_{h}\leq N_{h}, \ \ h=1,2,\dots, H\\
    n_{h}\in \mathbb{N}.
  \end{array}
\end{equation}
and correspondingly, for (\ref{fv1}):
\begin{equation}\label{mrsma}
    \begin{array}{c}
    \build{\min}{}{\mathbf{n}} \P\left(F(\mathbf{n}, \mathbb{S})\leq \tau_{1}\right)\\
    \mbox{subject to}\\
    \widehat{\Var}(\overline{y}_{_{ST}}^{j}(\mathbf{n})) + d_{j}^{+} - d_{j}^{-} = t_{j}, \quad j = 1,\dots,G, \\
    \mathbf{c}'\mathbf{n}+ c_{0}=C \\
    2\leq n_{h}\leq N_{h}, \ \ h=1,2,\dots, H\\
    n_{h}\in \mathbb{N},
  \end{array}
\end{equation}
for an specified aspiration level $\tau_{1}$.

Among the multiobjective techniques we find that the value function method is, in general, the
most commonly applied, its properties have been studied with most detail, see
\cite{rio89}, \cite{m99}, \cite{s86}, and the references therein. Note that the value function
$\mathbf{\phi}(\cdot)$ can be defined in an infinite number of forms, which represents a great
obstacle for its definition. Fortunately, some simple functions have given excellent results in
applications and they can be considered as promising approaches. One of these particular
forms is the weighting method. Under this approach, problem (\ref{fv}) can be expressed as:
\begin{equation}\label{fvwm}
    \begin{array}{c}
      \build{\min}{}{\mathbf{n}}\displaystyle\sum_{j=1}^{G}w_{j}\widehat{\Var}
      (\overline{y}_{_{ST}}^{j}),\\
       \mbox{subject to}\\
      \mathbf{c}'\mathbf{n}+ c_{0}=C \\
      2\leq n_{h}\leq N_{h}, \ \ h=1,2,\dots, H\\
      \V_{u}(\overline{\mathbf{y}}_{_{ST}}) \build{\rightarrow}{d}{} \mathcal{N}_{k}
      \left(\E\left(\V_{u}(\overline{\mathbf{y}}_{_{ST}})\right),
      \Cov\left(\V_{u}(\overline{\mathbf{y}}_{_{ST}})\right)\right)\\
      n_{h}\in \mathbb{N},
   \end{array}
\end{equation}
such that $\displaystyle\sum_{j=1}^{G}w_{j}=1, \quad w_{j}\geq 0\quad\forall\quad j=1,2, \dots
,G$; where the $w_{j}$'s weight the importance of each characteristic.

\section{Second approach: Stochastic multiobjective mathematical programming approaches}\label{sec6}

In this section, solutions for problem (\ref{equiv2}) are proposed  under diverse stochastic multiobjective
mathematical programming approaches. The properties of the solution obtained under the
different approaches are described in detail by \citet{k:63}, \citet{sm:84} and \citet{p:95}.

As shall be seen, each stochastic multiobjective mathematical programming approach can be stated
in several ways. In some cases, these possibilities are a consequence of assuming whether
$\widehat{\Var}(\overline{y}_{_{ST}}^{j}) \ j = 1, \dots, G,$ are correlated or not.

\subsection{Multiobjetive expected value solution, multiobjetive E-model}

Point $\mathbf{n} \in \mathbb{N}^{G}$ is the expected value solution to (\ref{equiv2}) if it is
an efficient solution in the Pareto sense to the following deterministic multiobjective mathematical
programming problem
\begin{equation}\label{solE}
    \begin{array}{c}
      \build{\min}{}{\mathbf{n}}\E\left(\V_{u}(\overline{\mathbf{y}}_{_{ST}})\right),\\
       \mbox{subject to}\\
      \mathbf{c}'\mathbf{n}+ c_{0}=C \\
      2\leq n_{h}\leq N_{h}, \ \ h=1,2,\dots, H\\
      n_{h}\in \mathbb{N}.
   \end{array}
\end{equation}

\subsection{Multiobjetive minimum variance solution, multiobjetive V-model}

The point $\mathbf{n} \in \mathbb{N}^{G}$ is the minimum variance solution to the problem
(\ref{equiv2}) if it is an efficient solution in the Pareto sense of the deterministic
multiobjetive mathematical programming problem
\begin{equation}\label{solV1}
  \begin{array}{c}
    \build{\min}{}{\mathbf{n}}
    \left(%
    \begin{array}{c}
      \Var\left(\widehat{\Var}(\overline{y}_{_{ST}}^{1})\right) \\
      \vdots\\
      \Var\left(\widehat{\Var}(\overline{y}_{_{ST}}^{G})\right) \\
    \end{array}%
    \right)\\
    \mbox{subject to}\\
    \mathbf{c}'\mathbf{n} + c_{0} = C  \\
    2\leq n_{h}\leq N_{h}, \ \ h=1,2,\dots, H\\
    n_{h}\in \mathbb{N},
  \end{array}
\end{equation}
This efficient solution is adequate if it is assumed that
$\widehat{\Var}(\overline{y}_{_{ST}}^{j}), \ j =1,\dots,G$ are uncorrelated, however if they are correlated
then  a more adequate approach is:

The point $\mathbf{n} \in \mathbb{N}^{G}$ is the minimum variance solution to the problem
(\ref{equiv2}) if it is an efficient solution in the Pareto sense of the deterministic matrix
mathematical programming problem
\begin{equation}\label{solV2}
    \begin{array}{c}
    \build{\min}{}{\mathbf{n}}
    \Cov
    \left(%
    \V_{u}(\overline{\mathbf{y}}_{_{ST}}))
    \right)\\
    \mbox{subject to}\\
    \mathbf{c}'\mathbf{n} + c_{0} = C  \\
    2\leq n_{h}\leq N_{h}, \ \ h=1,2,\dots, H\\
    n_{h}\in \mathbb{N}.
  \end{array}
\end{equation}
A general approach is studied by \citet{dgrr:11}.

\subsection{Multiobjetive expected value standard deviation solution, multiobjetive modified E-model}

Point $\mathbf{n} \in \mathbb{N}^{G}$ is an expected value standard deviation solution to the
problem (\ref{equiv2}) if it is an efficient solution in the Pareto sense of the mixed
deterministic multiobjetive-matrix mathematical programming problem
\begin{equation}\label{solV}
    \begin{array}{c}
    \build{\min}{}{\mathbf{n}}
    \left[
    \begin{array}{c}
      \E\left(\V_{u}(\overline{\mathbf{y}}_{_{ST}})\right) \\
      \left(\Cov \left(\V_{u}(\overline{\mathbf{y}}_{_{ST}}) \right)\right)^{1/2}
    \end{array}
    \right ]\\
    \mbox{subject to}\\
    \mathbf{c}'\mathbf{n} + c_{0} = C  \\
    2\leq n_{h}\leq N_{h}, \ \ h=1,2,\dots, H\\
    n_{h}\in \mathbb{N}.
   \end{array}
\end{equation}
where $\left(\mathbf{A}^{1/2}\right)^{2} = \mathbf{A}$, see \citet[Appendix]{mh:82}.

We now define the concept of efficient solution multiobjetive minimum risk of joint aspiration
level $\boldsymbol{\tau} = (\tau_{1}, \tau_{2}, \dots, \tau_{G})'$ and the efficient solution
with joint probability $\alpha$. Both solutions are obtained, respectively, by applying the
multivariate versions of minimum risk and Kataoka criteria, referred to in the literature as
criteria of maximum probability or satisfying criteria, due to the fact that, as we shall see,
in both cases the criteria to be used provide, in one way or another, ``good" solutions in
terms of probability, see \citet{k:63}.

\subsection{Multiobjetive minimum risk solution of joint aspiration level $\boldsymbol{\tau}$,
multiobjetive modified $\P$-model}

Point $\mathbf{n} \in \mathbb{N}^{G}$ is a minimum risk solution at joint aspiration level
$\boldsymbol{\tau}$ to the problem (\ref{equiv2}) if it is an efficient solution in the Pareto
sense of the multiobjetive stochastic mathematical programming problem
\begin{equation}\label{solMR1}
    \begin{array}{c}
    \build{\min}{}{\mathbf{n}}
    \left(%
    \begin{array}{c}
      \P\left(\widehat{\Var}(\overline{y}_{_{ST}}^{1}) < \tau_1\right) \\
      \vdots\\
      \P\left(\widehat{\Var}(\overline{y}_{_{ST}}^{G}) < \tau_G\right) \\
    \end{array}%
    \right)\\
    \mbox{subject to}\\
    \mathbf{c}'\mathbf{n} + c_{0} = C  \\
    2\leq n_{h}\leq N_{h}, \ \ h=1,2,\dots, H\\
    n_{h}\in \mathbb{N}.
  \end{array}
\end{equation}
It is also possible consider the follow alternative multiobjective $\P$-model
\begin{equation}\label{solMR2}
    \begin{array}{c}
    \build{\min}{}{\mathbf{n}}
    \P
    \left(%
    \begin{array}{c}
      \widehat{\Var}(\overline{y}_{_{ST}}^{1})<\tau_1 \\
      \vdots\\
      \widehat{\Var}(\overline{y}_{_{ST}}^{G})<\tau_G \\
    \end{array}%
    \right)\\
    \mbox{subject to}\\
    \mathbf{c}'\mathbf{n} + c_{0} = C  \\
    2\leq n_{h}\leq N_{h}, \ \ h=1,2,\dots, H\\
    n_{h}\in \mathbb{N}.
  \end{array}
\end{equation}
Again, (\ref{solMR2}) is more adequate if the response variables are correlated. However
(\ref{solMR2}) is sensibly more complicated to solve that (\ref{solMR1}). When $G = 2$,
\citet{pr:70} proposed an algorithm in a similar problem (probabilistic constrained
programming), which can be apply to solve (\ref{solMR2}).

\subsection{Multiobjetive Kataoka solution with probability $\alpha$}

Point $\mathbf{n} \in \mathbb{N}^{G}$ is a multiobjetive Kataoka solution with probability
$\alpha$ (fixed) to the problem (\ref{equiv2}) if it is an efficient solution in the Pareto
sense of the multiobjetive mathematical programming problem
\begin{equation}\label{solK1}
  \begin{array}{c}
    \build{\min}{}{\mathbf{n}, \boldsymbol{\tau}} \boldsymbol{\tau}\\
    \mbox{subject to}\\
      \P\left(\widehat{\Var}(\overline{y}_{_{ST}}^{j}) \leq \tau_{j}\right) = \alpha, \ j = 1, \dots, G\\
      \mathbf{c}'\mathbf{n} + c_{0} = C  \\
    2\leq n_{h}\leq N_{h}, \ \ h=1,2,\dots, H\\
    n_{h}\in \mathbb{N}.
  \end{array}
\end{equation}
Alternatively (\ref{solK1}) can be proposed as
\begin{equation}\label{solK2}
 \begin{array}{c}
    \build{\min}{}{\mathbf{n}, \boldsymbol{\tau}} \boldsymbol{\tau}\\
    \mbox{subject to}\\
      \P\left(%
    \begin{array}{c}
      \widehat{\Var}(\overline{y}_{_{ST}}^{1}) \leq \tau_{1} \\
      \vdots\\
      \widehat{\Var}(\overline{y}_{_{ST}}^{G}) \leq \tau_{r} \\
    \end{array}%
    \right) = \alpha\\
      \mathbf{c}'\mathbf{n} + c_{0} = C  \\
    2\leq n_{h}\leq N_{h}, \ \ h=1,2,\dots, H\\
    n_{h}\in \mathbb{N}.
  \end{array}
\end{equation}
Note that (\ref{solK1}) and (\ref{solK2}) are multiobjective probabilistic constrained
programs, see \citet{chc:63}, \citet{sm:84} and \citet{p:95}.

Many others approaches can be used to solve (\ref{equiv2}). For example, \citet{sm:84} propose
a stochastic version of a sequential technique termed Lexicographic method, for solving
(\ref{solV1}) and (\ref{solMR1}), or to applying  it directly to (\ref{equiv2}); among many others
options.

\section{Equivalent deterministic programs }\label{sec7}

In this section we study in detail several particular deterministic programs equivalent to
the stochastic multiobjective program (\ref{smmp}). Consider first the following remark: from
(\ref{ecyst}), (\ref{ccyst}), (\ref{sus}) and $j = 1, \dots, G$,
$$
    \widehat{\E}\left(\widehat{\Var}(\overline{y}_{_{ST}}^{j})\right) =  \sum_{h=1}^{H}\left(
     \frac{{{W_{h}}^{2}}}{n_{h}} - \frac{{W_{h}}}{N} \right) \frac{n_{h}}{n_{h}-1}s_{h j}^{2},
$$
$$
    \widehat{\Var}\left(\widehat{\Var}(\overline{y}_{_{ST}}^{j})\right) =\sum_{h=1}^{H}\left(
    \frac{{{W_{h}}^{2}}}{n_{h}} - \frac{{W_{h}}}{N} \right) \frac{n_{h}}{(n_{h}-1)^{2}}
    \left(m_{h j}^{4} - (s_{h j}^{2})^{2}\right),
$$
and
$$
  \widehat{m}_{h j}^{4} = \frac{1}{n_{h}}\sum_{i = 1}^{n_{h}}
    \left(y_{h i}^{j}- \overline{Y}_{h}^{j}\right)^{4}.
$$

\subsection{Multiobjective $V$-model}

From (\ref{solV1}) the multiobjective deterministic problem equivalent to (\ref{smmp}) via the $V$-model
is
\begin{equation}\label{esolV1}
  \begin{array}{c}
    \build{\min}{}{\mathbf{n}}
    \left(%
    \begin{array}{c}
      \widehat{\Var}\left(\widehat{\Var}(\overline{y}_{_{ST}}^{1})\right) \\
      \vdots\\
      \widehat{\Var}\left(\widehat{\Var}(\overline{y}_{_{ST}}^{G})\right) \\
    \end{array}%
    \right)\\
    \mbox{subject to}\\
    \mathbf{c}'\mathbf{n} + c_{0} = C  \\
    2\leq n_{h}\leq N_{h}, \ \ h=1,2,\dots, H\\
    n_{h}\in \mathbb{N}.
  \end{array}
\end{equation}
Which is solved by applying any of the multiobjective mathematical programming techniques.

\subsection{Multiobjective $\P$-model}

Proceeding as in \citet{dgrc:05}, the multiobjective deterministic problem equivalent to (\ref{smmp})
via the $\P$-model (\ref{solMR1}) is
\begin{equation}\label{dsolMR1}
    \begin{array}{c}
    \build{\min}{}{\mathbf{n}}
    \left(%
    \begin{array}{c}\displaystyle
      \frac{\tau_{1} - \widehat{\E}\left(\widehat{\Var}(\overline{y}_{_{ST}}^{1})\right)}{
        \sqrt{\widehat{\Var}\left(\widehat{\Var}(\overline{y}_{_{ST}}^{1})\right)}} \\
      \vdots\\
      \displaystyle
      \frac{\tau_{G} - \widehat{\E}\left(\widehat{\Var}(\overline{y}_{_{ST}}^{G})\right)}{
        \sqrt{\widehat{\Var}\left(\widehat{\Var}(\overline{y}_{_{ST}}^{G})\right)}} \\
    \end{array}%
    \right)\\
    \mbox{subject to}\\
    \mathbf{c}'\mathbf{n} + c_{0} = C  \\
    2\leq n_{h}\leq N_{h}, \ \ h=1,2,\dots, H\\
    n_{h}\in \mathbb{N}.
  \end{array}
\end{equation}

\subsection{Multiobjective Kataoka model}

From \citet{dgrc:05}, the multiobjective deterministic problem equivalent to (\ref{smmp}) via the
Kataoka model (\ref{solK1}) is given by
\begin{equation}\label{dsolMR2}
    \begin{array}{c}
    \build{\min}{}{\mathbf{n}}
    \left(%
    \begin{array}{c}
      \widehat{\E}\left(\widehat{\Var}(\overline{y}_{_{ST}}^{1})\right) + \Phi^{-1}(\delta) \
       \sqrt{\widehat{\Var}\left(\widehat{\Var}(\overline{y}_{_{ST}}^{1})\right)}\\
       \vdots\\
      \widehat{\E}\left(\widehat{\Var}(\overline{y}_{_{ST}}^{G})\right) + \Phi^{-1}(\delta) \
       \sqrt{\widehat{\Var}\left(\widehat{\Var}(\overline{y}_{_{ST}}^{G})\right)}\\
    \end{array}%
    \right)\\
    \mbox{subject to}\\
    \mathbf{c}'\mathbf{n} + c_{0} = C  \\
    2\leq n_{h}\leq N_{h}, \ \ h=1,2,\dots, H\\
    n_{h}\in \mathbb{N}.
  \end{array}
\end{equation}
where $\Phi$ denotes the distribution function of the standard Normal distribution.

Similar multiobjective deterministic problems equivalent to (\ref{smmp}) are obtained applying the
other stochastic solutions described in Section \ref{sec4}. Note that, if we combine each
stochastic solution with each multiobjective optmisation technique, we obtain an infinite
number of possible solutions of (\ref{smmp}). For example, note that the function of value
$f(\cdot)$ may take an infinite number of forms. One of these particular forms is the weighting
method, see \citet{rio89} and \citet{s86}. Under this approach, problem (\ref{dsolMR2}) can be
solved as:
\begin{equation}\label{WsolMR2}
    \begin{array}{c}
    \build{\min}{}{\mathbf{n}}
      \displaystyle\sum_{j=1}^{G} w_{j}\left\{\widehat{\E}\left(\widehat{\Var}(\overline{y}_{_{ST}}^{j})\right) + \Phi^{-1}(\delta) \
      \sqrt{\widehat{\Var}\left(\widehat{\Var}(\overline{y}_{_{ST}}^{j})\right)}\right\}\\
    \mbox{subject to}\\
    \mathbf{c}'\mathbf{n} + c_{0} = C  \\
    2\leq n_{h}\leq N_{h}, \ \ h=1,2,\dots, H\\
    n_{h}\in \mathbb{N}.
  \end{array}
\end{equation}
such that $\sum_{j=1}^{G}w_{j} = 1$, $w_{j} \geq 0$ $\forall$ $j= 1, 2, \dots, G$: where
$w_{j}$ weights the importance of each characteristic. This solution can be termed the
multiobjetive Kataoka-weighting solution with probability $\alpha$ equivalent to the problem
(\ref{smmp}), via the weighting method.

Similarly, the multiobjetive Kataoka solution with probability $\alpha$ to the problem
(\ref{equiv2}), via the goal programming is
\begin{equation}\label{GPsolMR2}
  \begin{array}{c}
    \build{\min}{}{\mathbf{n}}\displaystyle\sum_{j=1}^{G}w_{j}(d_{j}^{+} + d_{j}^{-}) \\
    \mbox{subject to} \\
    \widehat{\E}\left(\widehat{\Var}(\overline{y}_{_{ST}}^{j})\right) + \Phi^{-1}(\delta) \
    \sqrt{\widehat{\Var}\left(\widehat{\Var}(\overline{y}_{_{ST}}^{j})\right)} - d_{j}^{+} + d_{j}^{-} = \tau_{j}, \ j = 1,\dots,G,\\
    \mathbf{c}'\mathbf{n} + c_{0} = C  \\
    2\leq n_{h}\leq N_{h}, \ \ h=1,2,\dots, H\\
    n_{h}\in \mathbb{N}.
  \end{array}
\end{equation}
where
\begin{eqnarray*}
  d_{j}^{+} &=& \frac{1}{2} \left( \left|\widehat{\E}\left(\widehat{\Var}(\overline{y}_{_{ST}}^{j})\right) + \Phi^{-1}(\delta) \
        \sqrt{\widehat{\Var}\left(\widehat{\Var}(\overline{y}_{_{ST}}^{j})\right)} - \tau_{j}
        \right| \right .\\
   && \qquad\left. + \left(\widehat{\E}\left(\widehat{\Var}(\overline{y}_{_{ST}}^{j})\right) + \Phi^{-1}(\delta) \
        \sqrt{\widehat{\Var}\left(\widehat{\Var}(\overline{y}_{_{ST}}^{j})\right)} -
        \tau_{j}\right)\right), \\
  d_{j}^{-} &=& \frac{1}{2} \left( \left|\widehat{\E}\left(\widehat{\Var}(\overline{y}_{_{ST}}^{j})\right) + \Phi^{-1}(\delta) \
        \sqrt{\widehat{\Var}\left(\widehat{\Var}(\overline{y}_{_{ST}}^{j})\right)} -
        \tau_{j} \right| \right .\\
   && \qquad\left. - \left(\widehat{\E}\left(\widehat{\Var}(\overline{y}_{_{ST}}^{j})\right) + \Phi^{-1}(\delta) \
        \sqrt{\widehat{\Var}\left(\widehat{\Var}(\overline{y}_{_{ST}}^{j})\right)} -
        \tau_{j}\right)\right).
\end{eqnarray*}

Note that so far, the cost constraint $\displaystyle\sum_{h=1}^{H}c_{h}n_{h}+ c_{0}=C$
has been used in every stochastic multiobjective mathematical programming method. However, in several
situations, it can be used to represent existing restrictions for the availability of
man-hours for carrying out a survey, or restrictions on the total available time for performing
the survey, etc. These cases can be implemented by using the following constraint, see
\cite{ad81}:
$$
  \sum_{h=1}^{H}n_{h} = n.
$$

\section{Application}\label{sec8}

Consider data from a case study presented by \citet{aa81}. They study a forest survey conducted in Humbolt
County, California. The population was subdivided into nine strata on the basis of the timber
volume per unit area, as determined from aerial photographs. The two variables included in this
example are the basal area\footnote{In forestry terminology, 'Basal area' is the area of a
plant perpendicular to the longitudinal axis of a tree at 4.5 feet above ground.}, (BA), in square
feet, and the net volume in cubic feet (Vol.), both expressed on a per acre basis. The
variances, covariances and the number of units within stratum $h$ are listed in Table 1.

\begin{table}
\caption{\small Variances, covariances and the number of units within each stratum}
\begin{center}
\begin{footnotesize}
\begin{tabular}{ c r r r r }
\hline\hline
\multicolumn{2}{c}{} & \multicolumn{2}{c}{Variance} \\
\cline{3-4}
Stratum & $N_{h}$ & \hspace{.5cm} BA \hspace{.5cm} & \hspace{.5cm} Vol. \hspace{.5cm} & \hspace{.5cm}Covariance \\
\hline\hline
1 & 11 131 & 1 557 & 554 830 & 28 980 \\
2 & 65 857 & 3 575 & 1 430 600 & 61 591\\
3 & 106 936 & 3 163 & 1 997 100 & 72 369 \\
4 & 72 872 & 6 095 & 5 587 900 & 166 120\\
5 & 78 260 & 10 470 & 10 603 000 & 293 960 \\
6 & 51 401 & 8 406 & 15 828 000 & 357 300\\
7 & 24 050 & 20 115 & 26 643 000 & 663 300 \\
8 & 46 113 & 9 718 & 13 603 000 & 346 810\\
9 & 102 985 & 2 478 & 1 061 800 & 39 872 \\
\hline\hline
\end{tabular}
\end{footnotesize}
\end{center}
\end{table}

In this example, the stochastic multiobjective mathematical programming problem under approach
(\ref{smmp}) is
\begin{equation}\label{ej}
  \begin{array}{c}
  \build{\min}{}{\mathbf{n}}
    \left(%
    \begin{array}{c}
      \widehat{\Var}(\overline{y}_{_{ST}}^{1})\\
      \widehat{\Var}(\overline{y}_{_{ST}}^{2}) \\
    \end{array}%
    \right)\\
    \mbox{subject to}\\
    \displaystyle\sum_{h=1}^{9}n_{h}=1000 \\
    \V_{u}(\overline{\mathbf{y}}_{_{ST}}) \build{\rightarrow}{d}{} \mathcal{N}_{2}
    \left(\E\left(\V_{u}(\overline{\mathbf{y}}_{_{ST}})\right),
    \Cov\left(\V_{u}(\overline{\mathbf{y}}_{_{ST}})\right)\right)\\
    2\leq n_{h}\leq N_{h}, \ \ h=1,2,\dots, 9\\
    n_{h}\in \mathbb{N},
  \end{array}
\end{equation}
where $\E\left(\V_{u}(\overline{\mathbf{y}}_{_{ST}})\right)$ and
$\Cov\left(\V_{u}(\overline{\mathbf{y}}_{_{ST}})\right)$ are given by (\ref{ecyst}) and
(\ref{ccyst}) respectively.

\subsection{Particular solutions under the first approach}

Taking into account (\ref{ecyst}), (\ref{ccyst}) and (\ref{sus}), and from (\ref{emsma}) and
(\ref{fvwm}) the equivalent deterministic problem to (\ref{ej}) is
$$
    \begin{array}{c}
      \build{\min}{}{\mathbf{n}}
      \left\{\begin{array}{l}
               k_{1}\left(\displaystyle\sum_{j=1}^{2}w_{j}\widehat{\E}\left[\widehat{\Var}
               (\overline{y}_{_{ST}}^{j})\right]\right) \\
               \ + k_{2}\left(\sqrt{\displaystyle\sum_{j=1}^{2}w_{j}^{2}\widehat{\Var}\left[\widehat{\Var}
                (\overline{y}_{_{ST}}^{j})\right]+ 2 w_{1}w_{2} \widehat{\Cov}\left(\widehat{\Var}
                (\overline{y}_{_{ST}}^{1}),\widehat{\Var}
                (\overline{y}_{_{ST}}^{2})\right)}\right)
             \end{array}\right\} \\
       \mbox{subject to}\\
      \displaystyle\sum_{h=1}^{9}n_{h}=1000 \\
      2\leq n_{h}\leq N_{h}, \ \ h=1,2,\dots, 9\\
      n_{h}\in \mathbb{N},
   \end{array}
$$
where $w_{1} + w_{2}=1$,  and $k_{1} + k_{2}=1 \quad w_{j}, r_{j } (\mbox{fixed})\geq
0\quad\forall\quad j=1,2$, whose solution is termed: the \emph{weighting-modified $E-$model
solution}.

Similarly, from (\ref{mrsma}), the \emph{weighting-$P-$model solution} of (\ref{ej}) is obtained
solving the following equivalent deterministic problem
$$
    \begin{array}{c}
      \build{\min}{}{\mathbf{n}}
               \displaystyle\frac{\tau - \displaystyle\sum_{j=1}^{2}w_{j}\widehat{\E}\left[\widehat{\Var}
               (\overline{y}_{_{ST}}^{j})\right]}{
               \sqrt{\displaystyle\sum_{j=1}^{2}w_{j}^{2}\widehat{\Var}\left[\widehat{\Var}
                (\overline{y}_{_{ST}}^{j})\right]+ 2 w_{1}w_{2} \widehat{\Cov}\left(\widehat{\Var}
                (\overline{y}_{_{ST}}^{1}),\widehat{\Var}
                (\overline{y}_{_{ST}}^{2})\right)}}\\
       \mbox{subject to}\\
      \displaystyle\sum_{h=1}^{9}n_{h}=1000 \\
      2\leq n_{h}\leq N_{h}, \ \ h=1,2,\dots, 9\\
      n_{h}\in \mathbb{N},
   \end{array}
$$

\subsection{Particular solutions via the second approach}

From Section \ref{sec7}, using the weighting method, the \emph{multiobjective $V-$model-weighting model solution} of (\ref{ej}) is obtained from the
following equivalent deterministic problem
$$
    \begin{array}{c}
    \build{\min}{}{\mathbf{n}}
      \displaystyle\sum_{j=1}^{2} w_{j}\left\{\widehat{\Var}\left(\widehat{\Var}(\overline{y}_{_{ST}}^{j})\right)\right\}\\
    \mbox{subject to}\\
    \displaystyle\sum_{h=1}^{9}n_{h}=1000 \\
    2\leq n_{h}\leq N_{h}, \ \ h=1,2,\dots, 9\\
    n_{h}\in \mathbb{N}.
  \end{array}
$$
Analogously, the \emph{multiobjective $P-$model-weighting model solution} of (\ref{ej}) is
obtained by solving the following equivalent deterministic problem
$$
    \begin{array}{c}
    \build{\min}{}{\mathbf{n}}
      \displaystyle\sum_{j=1}^{2} w_{j}\displaystyle\frac{\tau_{j} - \widehat{\E}
      \left(\widehat{\Var}(\overline{y}_{_{ST}}^{j})\right)}{
      \sqrt{\widehat{\Var}\left(\widehat{\Var}(\overline{y}_{_{ST}}^{j})\right)}}\\
    \mbox{subject to}\\
    \displaystyle\sum_{h=1}^{9}n_{h}=1000 \\
    2\leq n_{h}\leq N_{h}, \ \ h=1,2,\dots, 9\\
    n_{h}\in \mathbb{N}.
  \end{array}
$$

Similarly, solving the following equivalent deterministic problem, the \emph{multiobjetive
Kataoka-weighting solution} of (\ref{ej}) is obtained
$$
    \begin{array}{c}
    \build{\min}{}{\mathbf{n}}
      \displaystyle\sum_{j=1}^{2} w_{j}\left\{\widehat{\E}\left(\widehat{\Var}(\overline{y}_{_{ST}}^{j})\right) + \Phi^{-1}(\delta) \
      \sqrt{\widehat{\Var}\left(\widehat{\Var}(\overline{y}_{_{ST}}^{j})\right)}\right\}\\
    \mbox{subject to}\\
    \displaystyle\sum_{h=1}^{9}n_{h}=1000 \\
    2\leq n_{h}\leq N_{h}, \ \ h=1,2,\dots, 9\\
    n_{h}\in \mathbb{N}.
  \end{array}
$$
In all cases, $\sum_{j=1}^{G}w_{j} = 1$, $w_{j} \geq 0$ $\forall$ $j= 1, 2$: where the $w_{j}$'s
weight the importance of each characteristic.

Table 2 shows the solutions obtained by some of the methods described in Sections \ref{sec5},
in particular, the weighting-modified $E-$model, weighting-$E-$model and
weighting-$V-$model are presented. Also, the solutions described in Section \ref{sec7},
are included, specifically the multiobjective $V-$model-weighting, $P-$model-weighting and Kataoka
model-weighting models. The first two rows in Table 2 include the optimum allocation for each characteristic,
BA and Vol. The last two columns show the minimum values of the
individual variances for the respective optimum allocations identified by each method. The
results were computed using the commercial software Hyper LINGO/PC, release 6.0, see
\citet{w95}. The default mathematical programming methods used by LINGO to solve the nonlinear
integer mathematical programming programs are Generalised Reduced Gradient (GRG) and
branch-and-bound methods, see \citet{bss06}. Some technical details of the computations are the
following: The maximum number of iterations of the methods presented in Table 2 was 2364
(multiobjective Kataoka-weighting solution) and the mean execution time for all the programs
was 2 seconds. Finally, observe that the greatest discrepancy found by the different methods
among the sizes of the strata occurred under the multiobjective $P-$model-weighting solution.
In part, this discrepancy is a consequence of the particular values for $\tau_{1}$ and $\tau_{2}$
in this method.

\begin{table}
\caption{\small Sample sizes and estimator of variances for the different allocations
calculated}
\begin{center}
\begin{minipage}[t]{400pt}
\begin{scriptsize}
\begin{tabular}{ c  c  c  c  c  c  c c c c c c}
\hline\hline Allocation & $n_{1}$ & $n_{2}$ & $n_{3}$ & $n_{4}$ & $n_{5}$ & $n_{6}$ & $n_{7}$ &
$n_{8}$ & $n_{9}$ & $\widehat{\Var}(\overline{y}_{_{ST}}^{1})$ &
$\widehat{\Var}(\overline{y}_{_{ST}}^{2})$ \\
\hline\hline
BA & 10 & 94 & 144 & 136 & 191 & 113 & 81 & 109 & 122 & 5.591 & 5441.105\\
Vol & 7 & 62 & 119 & 136 & 200 & 161 & 98 & 134 & 83 & 5.953 & 5139.531\\
\multicolumn{12}{c}{\textbf{First approach}\footnote{\scriptsize With $k_{1} = k_{2} = 0.5$ and $w_{1} = w_{2} = 0.5$.}} \\
Weighting-modified $E-$model
& 8 & 46 & 77 & 119 & 191 & 191 & 158 & 161 & 49 & 7.311 & 5593.494\\
Weighting-$E-$model
& 7 & 63 & 119 & 135 & 200 & 160 & 98 & 134 & 84 & 5.936 & 5139.645\\
Weighting-$V-$model
& 8 & 46 & 77 & 119 & 191 & 121 & 158 & 161 & 49 & 7.526 & 5997.963\\
\multicolumn{12}{c}{\textbf{Second approach}\footnote{\scriptsize Where $w_{1} = w_{2} = 0.5$.}} \\
Multiobjetive &&&&&&&&&&\\
$V-$model-weighting model
& 8 & 46 & 77 & 119 & 191 & 191 & 158 & 161 & 49 & 7.311 & 5593.494\\
Multiobjetive &&&&&&&&&&\\
$P-$model-weighting model\footnote{\scriptsize Where $\tau_{1} = 6$ and $\tau_{2} = 6000$.}
& 247 & 42 & 34 & 77 & 127 & 106 & 221 & 117 & 29 & 11.817 & 8980.960\\
Multiobjetive &&&&&&&&&&\\
Kataoka-weighting model\footnote{\scriptsize With $\Phi^{-1}(\delta) = 1.645$}
& 8 & 46 & 77 & 119 & 191 & 191 & 158 & 161 & 49 & 7.311 & 5593.494\\
\hline\hline
\end{tabular}
\end{scriptsize}
\end{minipage}
\end{center}
\end{table}

\section*{Conclusions}

It is important to stress that there is a potentially infinite number of possible solutions to
a stochastic multiobjective mathematical programming problem. Simply note that
such an infinite number due to the flexibility in choosing the value function. Therefore, this paper presents only
some few techniques of the area of multiobjective, stochastic and stochastic multiobjective
mathematical programming, which when combined, produce a number of possible solutions to the
problem of optimum allocation in multivariate stratified random sampling from a stochastic
multiobjective mathematical programming point of view.

Because of all these possibilities, it is difficult to set general rules for how to use these techniques,
leaving all responsibility on the skill of an expert in the area to determine which of
these approaches is best in a particular application.

As the reader can see, several solutions are given in the context of the application, and, for the sake of completeness,
the solutions in terms of goal programming are only indicated, see \citet{kea:08}.

\section*{Acknowledgments}

This research work was partially supported by IDI-Spain, Grants No. FQM2006-2271 and
MTM2008-05785, supported also by CONACYT Grant CB2008 Ref. 105657. This paper was written during J. A. D\'{\i}az-Garc\'{\i}a's stay as a visiting
professor at the Department of Probability Statistics of the Center of Mathematical Research,
Guanajuato, M\'exico.

\end{document}